\RequirePackage{amsmath}

\documentclass[11pt,a4paper]{article}

\pdfoutput=1

\usepackage{fullpage}

\usepackage{authblk}
\usepackage{algorithm, algorithmic}
\usepackage{times}
\usepackage{url}
\usepackage{paralist}
\usepackage{array}
\usepackage{float}
\usepackage{multirow}
\usepackage{amsmath}
\usepackage{amsthm}
\usepackage{longtable}
\usepackage{xcolor}
\usepackage{mathptmx}
\usepackage{graphicx} 
\usepackage{balance}
\usepackage{xspace}

\newcommand{\EE}{\ensuremath{\mathbb{E}}}
\newcommand{\A}{\mathcal{A}}
\newcommand{\para}[1]{\smallskip\noindent\textbf{#1}}
\newtheorem{defn}{Definition}
\newtheorem{claim}{Claim}
\newtheorem{lemma}{Lemma}
\newtheorem{theorem}{Theorem}

\usepackage{amsmath,amssymb,amsfonts,amsthm,dsfont}
\usepackage{hyperref}
\usepackage{cleveref}

\date{}

\begin{document}

\title{A Game Theoretical Analysis of Non-Linear Blockchain System}



\author[1]{Lin Chen}
\author[2]{Lei Xu}
\author[3]{Zhimin gao}
\author[1]{Ahmed Imtiaz Sunny}
\author[4]{Keshav Kasichainula}
\author[4]{Weidong Shi}

\affil[1]{
  Texas Tech University in Lubbock, Texas, United States of America\protect\\
  \texttt{chenlin198662@gmail.com,ahmed.sunny@ttu.edu}}
\affil[2]{
  University of Texas Rio Grande Valley in Edinburg,Texas United States of America\protect\\
  \texttt{xuleimath@gmail.com}}
\affil[3]{
  Auburn University of Montgomery in Montgomery,Alabama, United States of America\protect\\
  \texttt{mtion@msn.com}}
  \affil[4]{
  University of Houston in Houston,Texas, United States of America\protect\\
  \texttt{\{kkasichainula,wshi3\}@uh.edu}
  }
\maketitle









\begin{abstract}
    Recent advances in the blockchain research have been made in two important directions. One is refined resilience analysis utilizing game theory to study the consequences of selfish behavior of users (miners), and the other is the extension from a linear (chain) structure to a non-linear (graphical) structure for performance improvements, such as IOTA and Graphcoin. 
The first question that comes to mind is what improvements that a blockchain system would see by leveraging these new advances.
    In this paper, we consider three major properties for a blockchain system: $\alpha$-partial verification, scalability, and finality-duration. We establish a formal framework and prove that no blockchain system can achieve $\alpha$-partial verification for any fixed constant $\alpha$, high scalability, and low finality-duration simultaneously. We observe that classical blockchain systems like Bitcoin achieves full verification ($\alpha=1$) and low finality-duration, Ethereum 2.0 Sharding achieves low finality-duration and high scalability. We are interested in whether it is possible to  partially satisfy the three properties.

    
\end{abstract}





         





\maketitle

\section{Introduction}\label{sec-introduction}

Blockchain is a decentralized ledger that provides a system for self-interested parties to carry out transactions without a fully trusted central system. As such, there is no centralized party who can compute an optimal solution or a Nash equilibrium in advance and then tell each user to take certain actions, as this will violate the foundations of blockchain systems.
The basic idea behind this is that, instead of having trust in a centralized system or any other specific participant, each participant chooses to trust the majority of the participants and accepts the outcome achieved through consensus among them. \cite{nakamoto2008bitcoin}

One major reason that hinders the adoption of blockchain is scalability~\cite{vukolic2015quest}. For example, Bitcoin network can only process less than 10 transactions per second on average~\cite{francca2015homomorphic}, while typical payment systems like Visa can process thousands of transactions per second.


Recently, a variety of approaches are proposed to address the scalability issue. Most of them follow the general framework of divide and conquer, e.g., Zilliqa~\cite{zilliqa}, Harmony~\cite{harmony}, and Ethereum 2.0~\cite{ethereum20}, and use a sharding scheme that allows transactions to be processed by a subgroup of nodes (a sharding committee). A sharding scheme usually has a critical issue in terms of resilience, as the correctness of each transaction now solely depends on a subgroup of voters.  
Consequently, if common consensus protocols like Proof-of-Work (PoW) or Byzantine Fault Tolerance (BFT) is used within subgroups, then the fraction of malicious nodes within {\em every} subgroup cannot exceed 1/2 or 1/3, which is a significantly stronger assumption than that of a standard blockchain system. A typical blockchain system only requires that the fraction of malicious nodes do not exceed 1/2 or 1/3 of all nodes. We remark that both Harmony and Ethereum 2.0 claim that if subgroups are generated in a perfect randomized way, then the percentage of honest nodes within each subgroup is almost the same as their percentage in the whole group of nodes; however, this requires a perfect distributed random number generation as a separate procedure, which brings an additional assumption on the security of this additional procedure. 

To guarantee that every transaction is correctly executed by only relying on the standard assumption that the majority of the nodes are honest, we must require every transaction to be verified by all the nodes. 
Consequently, it is a straightforward question of whether scalability is achievable at all, as it appears that any divide and conquer based solution would inevitably reduce the total number of verifications received by a transaction. A "non-linear" blockchain structure recently introduced by IOTA~\cite{popov2016tangle} and Graphcoin~\cite{boyen2017blockchain}, can neglect such limitations. The basic idea is to allow blocks to be connected as a directed acyclic graph (DAG) instead of a chain. Such a non-linear structure implements a divide and conquer approach implicitly by allowing multiple blocks to be appended simultaneously, as a general graph can be extended in multiple directions. Meanwhile, if we treat different growing directions as soft forks or branches, then it is possible  (depending on system parameters) that they can ``merge" again in the future (see the following figure{~\ref{fig-dag}}, where the sequences of blocks that follow block $A$ and $B$ meet at block $C$). Therefore,  a node that tries to append a new block is required to verify a few previous blocks. So, there is a possibility that a block may still be verified by all the nodes, albeit the delay of such verification.

On a high level, there are three crucial properties involved in a general blockchain system: verification, scalability, and finality-duration. 
In a nutshell, $\alpha$-partial verification requires every transaction to be verified by at least $\alpha$ fraction of all the nodes (which thus ensures resilience under the standard assumption that the majority of the nodes follow the protocol); scalability 
means the system throughput, or the total number of transactions executed per unit of time, is proportional to the total number of participating nodes; and finality-duration 
means the delay in reaching consensus on the correctness of the execution of each transaction. We give a precise definition in Section~\ref{sec-problem}.

Classical blockchain systems like Bitcoin achieves full verification and low finality-duration, but not scalability. This is because Bitcoin requires every block, and hence the transactions within a block, to be verified by all the nodes; meanwhile, it has a constant finality-duration because every block is finalized after a constant number of blocks 
are appended afterward. However, it does not scale, as the increase in the number of nodes does not allow the system to handle more transactions per unit of time, which has been pointed out in many prior papers \cite{Solidity_ether,nakamoto2008bitcoin,sasson2014zerocash,SompolinskyLZ16}.On the other hand, blockchain systems like { Harmony~\cite{harmony} and Ethereum 2.0~\cite{ethereum20}} achieve constant finality-duration and scalability, but not full verification. For example, the sharding scheme used in Ethereum 2.0 allows a block to be verified within a shard (which is a subset of nodes). 

 In this paper, we provide a view on the relationships between full verification (or more precisely, $\alpha$-partial verification for any constant $\alpha$), low finality-duration, and scalability. More precisely: 

\para{Our contributions.} 
{We show that it is impossible to achieve full verification, low finality-duration and scalability simultaneously.}

Given the fact that: 
\begin{inparaenum}[(i)]
	\item Bitcoin achieves full verification and low (asymptotically constant) finality-duration, but not scalability; and 
	\item Ethereum 2.0 achieves low (constant) finality-duration and scalability, but not full verification,
\end{inparaenum}
it is natural to explore to what extent a blockchain system can be designed considering the different trade-off scenarios.  In particular, does there exist a blockchain system that satisfies both full verification and scalability/finality-duration? Is it possible to have a system that partially satisfies all of the three properties? We give an affirmative answer in this paper. In particular, we prove that by adopting a non-linear blockchain system 
and employing a game-theoretical analysis, we can construct a system which achieves full verification and a trade-off between scalability and finality-duration. Informally speaking, the following properties hold simultaneously for the constructed blockchain system: 
\begin{enumerate}[(i)]
    \item $O(s)$ new blocks are generated per unit of time on average;
    \item after $O(s\log s)$ units of time, with a very high probability, each block will be verified by all users in the systems.
\end{enumerate}
Here $s$ is a system parameter that can be set suitably at the genesis block. When $s=1$, the non-linear system degenerates to a linear system with a fixed block generation rate that is independent of the nodes in the system, while the delay which is an indication of finality duration, is a constant. This coincides with the classical Bitcoin system. Conversely, $s$ can be as high as $O(m)$ where $m$ is the number of nodes. In this case, the system is fully scalable, albeit that only a sufficiently long delay ($O(s\log s)$) can ensure full verification. However, if we set $s$ to be $O(m)$ to enforce scalability and meanwhile enforce the delay to be some constant instead of $O(s\log s)$, then full verification cannot be guaranteed. 



We remark that the big-$O$ notation in our statements hides a constant which is roughly the average time for a block to be generated, that is, we measure the delay in terms of the number of blocks; therefore, Bitcoin is considered as low (asymptotically constant) finality-duration as the delay is constant blocks. Our result does not conflict with prior researches that complain about the ``high'' finality-duration of Bitcoin because of the long time it takes to generate a single block \cite{popov2016tangle, boyen2017blockchain}. The research that tries to decrease such a block generation time is parallel to this paper. For example, if a lighter version of PoW can be used in the existing Bitcoin system, then it can also be used directly in our non-linear blockchain system, while our impossibility result, as well as the trade-off between finality-duration and scalability, remain the same.

\section{Related Work}
The study of e-cash systems dates back to 1983~\cite{chaum1983blind,sander1999auditable}. However, all such systems require a centrally or quasi-centrally controlling authority. 
A well-known exception, Bitcoin, was introduced by Nakamoto~\cite{nakamoto2008bitcoin} in 2008, which uses a public ledger known as a {\em blockchain} to record transactions carried out between users. Following this line of research, various alternative blockchain-based transaction systems are proposed~\cite{Solidity_ether,miers2013zerocoin,sasson2014zerocash,SompolinskyLZ16}, further improving the performance and security of Bitcoin as well as extending the system to deal with applications beyond transactions (e.g., smart contracts). {In \cite{sompolinsky2015secure}, Sompolinsky and Zohar have introduced an alternative to the longest chain that allows more transactions to take place at a lower cost.} { Recently,  Sompolinsky et al. \cite{sompolinsky2020phantom} have shown faster block generation by generalizing blockchain to a direct acyclic graph of blocks.  Blockchain-based consensus protocols Fantômette and Avalanche that rely on blockDAG were proposed in \cite{azouvi2018betting,rocket2018snowflake}.}  We refer the readers to several surveys on blockchain systems~\cite{tschorsch2016bitcoin,conti2018survey,khalilov2018survey,salman2018security,ali2019applications,yang2019integrated,liu2019survey,bano2019sok,garay2020sok}. In particular, \cite{tschorsch2016bitcoin} provides a comprehensive introduction to the bitcoin network,{\cite{bano2019sok,garay2020sok} focus on the systematized study of the blockchain consensus protocols},
\cite{conti2018survey,khalilov2018survey,salman2018security} focus on the security and privacy results on blockchain, \cite{ali2019applications} focuses on the applications of blockchain. The most relevant survey to this paper is \cite{liu2019survey}, which summarizes recent results on game-theoretical studies of blockchain.  However, most of the existing game theoretical research primarily focuses on the traditional linear blockchain system, only a very recent paper by Popov et al.~\cite{popov2019equilibria} gives the first game-theoretical analysis of IOTA. Their result, however, does not establish the trade-off between scalability and finality-duration.

\subsection{Classical and Non-linear Blockchain}\label{subsec:background}
\para{Chain-structured blockchain.}
Most of the existing blockchain systems, e.g., Bitcoin, Ethereum, Hyperledger, follow the classical structure where blocks form a chain as illustrated by Fig~\ref{fig-linear}. 

\begin{figure}[h!]
    \centering
	\includegraphics[width=0.50\textwidth]{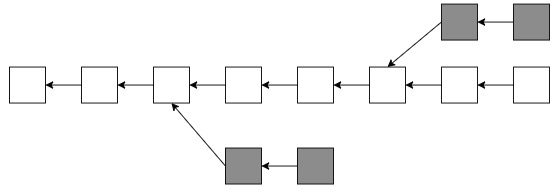}
	\caption{Chain-structured blockchain. The chain grows from left to right. White squares form the main chain, and gray squares form the side chains that are discarded eventually.}
	\label{fig-linear}
\end{figure}
\para{Non-linear (graph-structured) blockchain.}
Popov introduced the concept of {\em tangle}~\cite{popov2016tangle} which allows a blockchain to adopt a directed acyclic graph (DAG) architecture. We summarize the abstract model of a non-linear blockchain in Section~\ref{sec-problem}. We briefly review IOTA, which is the most well-known non-linear blockchain system so far. On a high level, IOTA allows each transaction be an individual node linked in the distributed ledger. We may interpret a transaction as a block in such a system. In the tangle, each user needs to select one transaction from the pool as well as two previous blocks (transactions) in the system. The user verifies these two transactions and mines a new block referring to them. Then this new block (transaction) is broadcasted to the tangle network.  \figurename~\ref{fig-dag} gives a simple example of a non-linear blockchain. 

\begin{figure}[h!]
    \centering
	\includegraphics[width=0.50\textwidth]{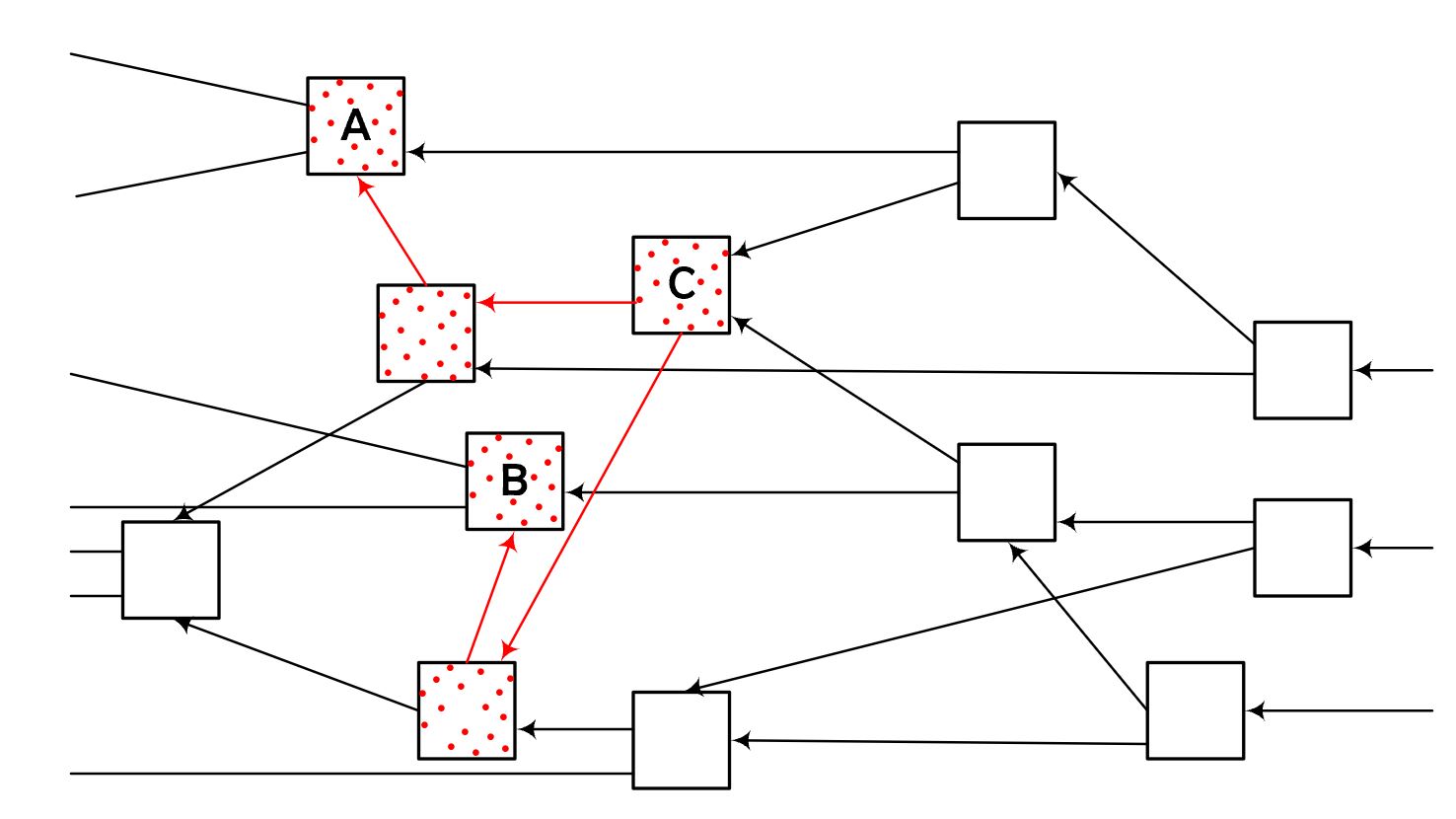}
	\caption{A non-linear blockchain. White squares are verified transactions/blocks.}
	\label{fig-dag}
\end{figure}




\section{The Abstract Model}\label{sec-problem}
We describe an abstract model of a non-linear blockchain which is general enough to incorporate existing well-known non-linear blockchain systems like IOTA and Graphcoin.

A non-linear blockchain \textit{NLB} is defined by a quadruple 
$$\textit{NLB}: (\mathcal{A}, \mathcal{C}, \mathcal{R}, \mathcal{E}), \text{where}$$ 
\begin{compactitem}
    \item $\mathcal{A}$ defines the rules of building and adding a new block to the blockchain. Since we are considering non-linear blockchain, $\mathcal{A}$ allows multiple blocks to be added simultaneously. 
    \item $\mathcal{C}$ defines the way to check a block, including validity verification, such as whether the block has the correct format and whether transactions included in the block are valid, and whether the block is finalized.
    \item $\mathcal{R}$ defines the way how the award is assigned to a user who adds a new block to the DAG. A \textit{NLB} needs to encourage users to participate in the construction of the blockchain by giving rewards to those who add new blocks.
    \item $\mathcal{E}$ defines the rules to eliminate conflicting blocks. Similar to a linear blockchain, it is possible that multiple participants have different local copies of the blockchain, and $\mathcal{E}$ determines which version should be kept.
\end{compactitem}

{Next, we provide formal definitions of the three metrics of a blockchain system that we mentioned earlier.}

\begin{defn}[Partial verification]
	For any fixed $\alpha\in (0,1]$, a blockchain system satisfies the property of $\alpha$-partial verification if every block is verified by at least $\alpha$ fraction of the total nodes in the system before it is finalized.
\end{defn}
In particular, if $\alpha=1$, then each block is verified by all the nodes and we call it {\it full verification}.
If a blockchain system satisfies full verification, then resilience follows directly from standard assumptions on the percentage of honest nodes among all nodes, e.g., if the blockchain uses PoW or BFT as the consensus protocol, then it is resilient, once the majority or 2/3 of nodes follow the protocol. This is also true for $\alpha$-partial verification if the $\alpha$-fraction of the nodes are randomly selected from all nodes. Indeed, many recently developed blockchain systems that claim to achieve scalability, e.g., Ethereum 2.0~\cite{ethereum20} implements such an idea by letting a subset of nodes (i.e., a {\it Shard}) compute and verify a smart contract. The size of such a subset divided by the total number of nodes gives the percentage $\alpha$. 

\begin{defn}[Scalability]
	The throughput of a blockchain system is the number of blocks $n_b$ that can be added to the system in a fixed time. A blockchain system scales with the number of nodes $m$ in the system if $n_b\to \infty$ when $m\to\infty$. Particularly, a blockchain system fully scales with the number of nodes $m$ if $n_b=\Omega(m)$. 
\end{defn}
It should be clear that the definition of scalability or full scalability does not depend on the length of the time period chosen for throughput. It captures the possibility of speeding up blockchain generation with more participating nodes; consequently, classical blockchain systems like Bitcoin does not scale.  

\begin{defn}[Finality-duration]
	The finality-duration of a blockchain system is the time difference between the time point when a block is appended and the time point when a block receives full verification.
\end{defn}
We say the finality-duration of a blockchain system is low (or constant) if the finality-duration is independent of the nodes in the system; consequently, classical blockchain systems like Bitcoin has a low (asymptotically constant) finality-duration because after a fixed number of blocks are appended, all the nodes start following the main chain, thus blocks on the main chain will receive full verification.


\section{Impossibility Result}
\begin{theorem}\label{thm:imp}
There does not exist a blockchain system that simultaneously satisfies
\begin{inparaenum}[(i)]
    \item scalability;
    \item low finality-duration; and
    \item $\alpha$-partial verification for an arbitrary constant $\alpha\in (0,1]$.
\end{inparaenum}
\end{theorem}
The proof follows from a counting argument on the total number of verifications. 
\begin{proof}
    Suppose, on the contrary, that there exists such a block-chain. Then by definition, every block or transaction will receive verifications from at least $\alpha$ fraction of the nodes within a constant delay. Let $c_0$ be the constant delay. Consider an arbitrary node $x$ and let $\tau_x$ be the fixed time it takes for node $x$ to perform one verification. Let the throughput of the blockchain be $n_b$, then by definition of scalability, $n_b=n_b(m)\to \infty$ when $m\to\infty$. Note that all the $n_b$ blocks generated shall be verified by at least $\alpha$ fraction of nodes within the delay of $c_0$, which means on average, every node should perform $\alpha n_b$ verifications within $c_0$. However, node $x$ can only perform $c_0/\tau_x$ verifications, which is a constant. Since $n_b\to\infty$, when $m$ is sufficiently large, $\alpha n_b>c_0/\tau_x$ for any fixed $\alpha$. Therefore, it is impossible for an arbitrary node $x$ to complete all the verifications. Hence, the three properties, scalability, low finality-duration and $\alpha$-partial verifications, cannot be satisfied simultaneously.
\end{proof}
\noindent\textbf{Remark.} If $\alpha$ is not a constant, e.g., $\alpha=10/N$ where $N$ is the total number of nodes, then $\alpha\to 0$ and it is indeed possible to guarantee scalability and low finality-duration simultaneously.

\section{Satisfying Impossibility-Triangle in a Liquid Way}
As we have mentioned before, Bitcoin achieves full verification and low finality-duration at the cost of scalability, Ethereum 2.0 achieves scalability and low finality-duration, but cannot guarantee a constant $\alpha$ for partial verification. These systems obey our impossibility triangle by conforming to two properties and disregarding the third one. But what if we want all of them in partial? More precisely, can we have a more fine-grained system that can can smoothly transform from one extreme to another by controlling a simple system parameter? In this section, we construct a non-linear blockchain system, which builds upon the basic structure of IOTA, and show that it achieves full verification, and meanwhile has a trade-off between scalability and finality-duration. In particular, a larger scalability implies longer finality-duration, and this can be controlled through a single system parameter $\Delta$ as we will define later. 

\subsection{Non-linear Blockchain (NLB) Construction}
We first propose a concrete construction of NLB that achieves both security and scalability under the agent model (Based on the agent model, every participant is an agent who tries to maximise his/her profit).
Without loss of generality, we assume that each block only includes one transaction. In the following, we broaden the terms and use them interchangeably.
We first define some concepts. 
\begin{defn}[Block distance, descendant, and ancestor]
	Given two blocks $A$ and $B$, we define the distance between the two blocks as the length of the shortest directed path from $A$ to $B$, which is denoted as $d(A,B)$. If there is no such a directed path, we define $d(A,B)=\infty$. 
	If $0<d(A,B)<\infty$, we say $B$ is a descendant of $A$, and $A$ is an ancestor of $B$.
	For a block $B$ and each $1\le k\le \ell$, let $\textit{Anc}(B,k)=\{X|d(X,B)=k\}$ and $\textit{Des}(B,k)=\{d(B,X)=k\}$, where $\ell$ is a given parameter.
\end{defn}

The new NLB is constructed as follows:
\begin{itemize}
    \item $\mathcal{A}$. The new \textit{NLB} assumes that there is a pool of new transactions from which a user can select one to construct a new block, which refers to two previous blocks\footnote{Our analysis in this paper also works if a new block refers to any fixed constant (greater than or equal to 2) number of blocks. For ease of presentation, we take this number to be 2 throughout this paper.}. The user then does lightweight mining to fix this information in the newly constructed block. 
    {Lightweight mining is a common approach used in blockchain systems supporting lightweight users (see, e.g.~\cite{xu2017epbc}). It means that the user needs shorter computation time, compared with the standard mining process, to find a value for the block that makes it a valid block. This is usually done by a loose requirement on the hash result of the block together with the mined value. 
    }
    Suppose that the newly built block is $B$, the user also verifies blocks in $\textit{Anc}(B,k), 1 \leq k \leq \ell$, where $\ell$ is a pre-defined system parameter that determines how many previous blocks the producer of a new block should verify.  
	\item $\mathcal{C}$. To check a block $B$, the algorithm first checks whether the block format is correct, including the verification of the mining outcome. The algorithm also checks whether $B$ is finalized or not, which is determined by 
	$$\textit{num}_\textit{Anc} \leftarrow |\cup_{k=1}^\ell \textit{Anc}(B,k)|.$$
	If $\textit{num}_\textit{Anc}$ is larger than the system pre-defined threshold, $B$ is finalized.
	\item $\mathcal{R}$. Each block has a reward value and the system imposes an upper bound on the maximal reward offered by a transaction, so that the largest and smallest reward among transactions (and blocks) can differ by a factor at most $\Delta$. The producer of the new block also receives rewards from previous blocks. Specifically, each block is associated with a uniform verification cost $\textit{vrf}$, which is divided into $\ell$ parts such that $\textit{vrf}_1 < \textit{vrf}_2 < \cdots < \textit{vrf}_\ell$ and $\sum_{k=1}^\ell \textit{vrf}_k = \textit{vrf}$. For each $1 \leq k \leq \ell$, the producer of block $B$ gets reward $\textit{vrf}_k / |\textit{Des}(X,k)|$ for each $X \in \textit{Anc}(B,k)$. This means that the verification reward of $\textit{vrf}_k$ from block $X$ is evenly distributed among all descendants in $\textit{Des}(B,k)$. Note that the reward is not uniformly divided and it will be only collected when the new block is finalized.
	\item $\mathcal{E}$. The constructed NLB adopts the largest-weighted descendants principle (LWD) to eliminate disagreement, i.e., when there are blocks containing conflicting transactions, the one whose descendants have a largest total weight will be selected. Note that this will not prevent multiple non-conflict blocks from being added in parallel. 
	Formally, for each block $B$, let $\textit{Des}(B)=\cup_{h\ge 1} \{X|d(B,X)\}$. If there are two conflicting blocks $B$, $B'$ and $|\textit{Des}(B)|>|\textit{Des}(B')|$, then $B$ prevails, that is, users will abandon $B'$ together with all its descendants in the sense that a new block will not refer to any of these blocks. 
\end{itemize}

Note that when considering the private costs for the mining task of the agents, the only difference is that the profit of each block is no longer its reward, instead, the new profit of each agent should be the reward of each task subtracting the (agent-dependent) cost.
All results in our paper still hold if the largest profit  and smallest profit differ by at most $\Delta$ times for all agents. In our paper, the blockchain system is designed such that the largest and smallest reward can differ by at most $\Delta$ times. Incorporating costs of agents, however, this cannot be ensured by system design. Specifically, if an agent has an excessive private cost then the ratio can be unbounded. But in practice, it is plausible to assume that the private cost is usually small compared with the reward of blocks.


\subsection{Scalability and Finality-duration Analysis}

We first give a high-level summary of the workflow of the proposed NLB system. Transactions are generated over time and form a pool. Each transaction is associated with a distinct transaction reward and a fixed verification reward $\textit{vrf}$. Each time, a miner will select one transaction from the pool and append a block, which refers to two previous blocks. Here the miner needs to decide two things:
\begin{inparaenum}[(i)]
	\item which transaction to include, and 
	\item which two previous blocks to refer to. 
\end{inparaenum}
As we assume that miners are rational players, they will strategically make their decisions to maximize their profits, and this section is devoted to analyze the scalability and security of the system under an arbitrary Nash equilibrium.

We formalize the problem as follows. Let the pool consist of $n$ transactions, with the transaction reward being $p_1,p_2,\cdots,p_n$. Let $m$ be the number of miners, with computational powers being $u_1,u_2,\cdots,u_m$. As we mentioned, each miner will mine a new block by including one transaction from the pool. If multiple miners say, miners in the subset of $S$, all choose the same transaction, then they compete, and only one of them will succeed, and the probability that some miner $i\in S$ succeeds is $\frac{u_i}{\sum_{h\in S}u_h}$. If, however, all miners choose different transactions, then each of them can append a new block. 
In the following section, we will analyze the scalability and finality-duration of the constructed NLB separately.

\subsubsection{Scalability}\label{sec-solution}
For scalability, we are interested in how many different transactions from the pool can be selected by the miners simultaneously. Note that the more different transactions are chosen, the higher scalability is. When miners choose transactions simultaneously, we are considering the worst-case because if miners are selecting transactions at different times, later ones may be able to avoid conflicts with earlier ones. Let $n$ be the number of available transactions and $m$ be the number of miners, the following Theorem~\ref{thm:poa-ub} implies that the system is scalable even in the worst case such that when there are sufficiently many transactions, the throughput will be $O(m/\Delta)$ where $\Delta$ is a system parameter part of $\mathcal{R}$. By controlling $\Delta$, we can control the scalability of the system. In particular, when we set $\Delta$ to be a constant, the system becomes fully scalable with the number of nodes.  In the remainder of this section is devoted to proving Theorem~\ref{thm:poa-ub}. 
\begin{theorem}\label{thm:poa-ub}
	With probability at least $1-\max\{e^{-\Theta({m})},e^{-\Theta({n})}\}$, the number of blocks mined by $m$ miners in an arbitrary Nash equilibrium is at least $\min\{c_1m/\Delta,c_2n\}$ for some universal constants $c_1,c_2$.
\end{theorem}

Notice that a Nash equilibrium always exists by allowing mixed strategies~\cite{nash1951non}.
Towards the proof, we introduce some notations. For simplicity, let all the transaction rewards be $p_1\ge p_2\ge\cdots\ge p_n$. By the design of our system we require that $p_1/p_n\le \Delta$. Note that the strategy of a miner is to select one transaction. We consider the general mixed strategy of a miner where he/she can specify a probability for each transaction. 

Consider an arbitrary Nash equilibrium and let $\pi^{(i)}=(\pi^{(i)}_1,$ $\pi^{(i)}_2,\cdots,\pi^{(i)}_n)$ be the strategy of miner $i$ in the equilibrium, where $\pi^{(i)}_j$ is the probability that he chooses transaction $j$. It is obvious that $\sum_{j=1}^n\pi^{(i)}_j=1$ for any $1\le i\le m$. Let $X^{(i)}_j$ be the $0$-$1$ random variable that indicates whether miner $i$ chooses transaction $j$. Then $X^{(i)}_j=1$ with probability $\pi^{(i)}_j$ and $X^{(i)}_j=0$ with probability $1-\pi^{(i)}_j$.

Consider the above Nash equilibrium. Intuitively, if only a small number of transactions are selected, then miners must have devoted their probabilities to a few transactions. Therefore, to show that a sufficient number of distinct transactions are selected in expectation, by miners, we need to show that the miners are distributing their probabilities in a fair way among transactions, as is implied by the following lemma.

\begin{lemma}\label{lemma:individual-bound}
If there exists some transaction $j_1$ such that $\sum_{i=1}^m \pi^{(i)}_{j_1} \ge  12\Delta,$
	then for every transaction $j$, it holds that $\sum_{i=1}^m \pi^{(i)}_{j}\ge  1/2$.
\end{lemma}
\begin{proof} Suppose, on the contrary, the lemma is not true, that is, there exists some transactions $j_1$ and $j_2$ such that $\sum_{i=1}^m \pi^{(i)}_{j_1}\ge  12\Delta$ and $\sum_{i=1}^m \pi^{(i)}_{j_2}< 1/2$. Consider the set of miners that choose transaction $j_1$ with positive probability. For simplicity, let these miners be miner $1,2,\cdots,k$ such that $u_1\ge u_2\ge\cdots\ge u_k$. We show in the following that miner $k$ can change his strategy to get a strictly higher profit, contradicting the fact that this is a Nash equilibrium, and consequently, the lemma is proved. More precisely, we argue that player $k$ can get strictly larger profit (in expectation) by increasing his probability of choosing transaction $j_2$ and meanwhile decreasing his probability of choosing $j_1$.
	
		The expected profit that miner $k$ can get from transaction $j_1$ and $j_2$ using his current strategy is equal to
	$$p_{j_1}\EE[\Gamma_1]+p_{j_2}\EE[\Gamma_2],$$
	where for $h=1,2$, we have
	\[
	\Gamma_h=\left\{
	\begin{array}{cc}
	0,\quad\textrm{if } \sum_{i=1}^nu_iX^{(i)}_{j_h}=0\\
	\frac{u_kX_{j_h}^{(k)}}{\sum_{i=1}^m u_iX^{(i)}_{j_h}}, \quad \textrm{Otherwise.}
	\end{array}
	\right.
	\]	
	
	If $k$ changes his strategy by choosing $j_1$ with the probability of $0$ and choosing $j_2$ with the probability of $\pi^{(k)}_{j_1}+\pi^{(k)}_{j_2}$, then the expected profit he can get from $j_1$ and $j_2$ is equal to
	$p_{j_2}\EE[\tilde{\Gamma}_2],$
	where
	\[
	\tilde{\Gamma}_2=\left\{
	\begin{array}{cc}
	0,\quad\textrm{if } \sum_{i\neq k}u_iX^{(i)}_{j_2}+u_k\tilde{X}^{(k)}_{j_2}=0\\
	\frac{u_k\tilde{X}_{j_2}^{(k)}}{\sum_{i\neq k}u_iX^{(i)}_{j_2}+u_k\tilde{X}^{(k)}_{j_2}}, \quad \textrm{Otherwise.}
	\end{array}
	\right.
	\]	
	and $\tilde{X}_{j_2}^{(k)}$ is the $0$-$1$ random variable that takes the value $1$ with the probability of $\pi^{(k)}_{j_1}+\pi^{(k)}_{j_2}$.		In the following we show that
	$$p_{j_2}\EE[\tilde{\Gamma}_2-\Gamma_2]> p_{j_1}\EE[\Gamma_1],$$
	which implies the correctness of the lemma. We prove the following two claims. 
	\begin{claim}\label{claim:1}
		$\EE[\Gamma_1]< \frac{\pi_{j_1}^{(k)}}{2\Delta}.$
	\end{claim}
	
	\begin{proof}
	Let $X=\sum_{i\neq k} X_{j_1}^{(i)}$ and $\mu=\EE[X]\ge 5\Delta/\theta-\EE[X_{j_1}^{(k)}]=6\Delta/\theta-\pi_{j_1}^{(k)}$. For any $\delta\in (0,1)$, we have
	\begin{eqnarray*}
		\EE[\Gamma_1]=&&\EE[\Gamma_1|X> (1-\delta)\mu]\Pr(X>(1-\delta)\mu)\\&&+\EE[\Gamma_1|X\le  (1-\delta)\mu]\Pr(X\le (1-\delta)\mu)
	\end{eqnarray*}
	Given that $\Gamma_1\le 1$, we know that
	\begin{eqnarray*}
		&&\EE[\Gamma_1|X\le  (1-\delta)\mu]\\
		&=&\EE[\Gamma_1|X\le  (1-\delta)\mu, X_{j_1}^{(k)}=1]\Pr(X_{j_1}^{(k)}=1)\\&&+\EE[\Gamma_1|X\le  (1-\delta)\mu, X_{j_1}^{(k)}=0]\Pr(X_{j_1}^{(k)}=0)\\
		&=& \pi_{j_1}^{(k)}\cdot\EE[\Gamma_1|X\le  (1-\delta)\mu, X_{j_1}^{(k)}=1]\le \pi_{j_1}^{(k)}.
	\end{eqnarray*}
	Meanwhile, by $u_1\ge u_2\ge \cdots\ge u_k$, we have
	$$\Gamma_1\le \frac{X_{j_1}^{(k)}}{\sum_{i=1}^k X_{j_1}^{(i)}}.$$
	According to Chernoff bound, we know that
	$$\Pr(X\le (1-\delta)\mu)\le e^{\frac{-\delta^2\mu}{2}}.$$	
	
	Therefore,
	\begin{eqnarray*}
		\EE[\Gamma_1]&=&\EE[\Gamma_1|X> (1-\delta)\mu]\Pr(X>(1-\delta)\mu)\\&&+\EE[\Gamma_1|X\le  (1-\delta)\mu]\Pr(X\le (1-\delta)\mu)\\
		&\le & \EE[\frac{X_{j_1}^{(k)}}{(1-\delta)\mu+X_{j_1}^{(k)}}]+ e^{\frac{-\delta^2\mu}{2}}\cdot \pi_{j_1}^{(k)}
	\end{eqnarray*}
	Consider the function $\varphi(x)=\frac{x}{(1-\delta)\mu+x}$, it is easy to verify that $-\varphi(x)$ is convex in $x\in (0,+\infty)$, therefore, by Jensen's inequality 
	\begin{eqnarray*}
	\frac{\EE[X_{j_1}^{(k)}]}{(1-\delta)\mu+\EE[X_{j_1}^{(k)}]}=\varphi(\EE[X_{j_1}^{(k)}])&\ge& \EE[\varphi(X)]\\&=&\EE[\frac{X_{j_1}^{(k)}}{(1-\delta)\mu+X_{j_1}^{(k)}}].
	\end{eqnarray*}
	Now consider the function $xe^{-x}$. It is easy to verify that the function decreases when $x\ge 1$, therefore $e^{-x}\le \frac{1}{ex}$ for $x\ge 1$, hence, for $\frac{\delta^2\mu}{2}\ge 1$, we have
	\begin{eqnarray*}
	\EE[\Gamma_1]&\le& \frac{\EE[X_{j_1}^{(k)}]}{(1-\delta)\mu+\EE[X_{j_1}^{(k)}]}+e^{\frac{-\delta^2\mu}{2}}\cdot\pi_{j_1}^{(k)}\\&\le& \frac{\pi^{(k)}_{j_1}}{(1-\delta)\mu+\pi^{(k)}_{j_1}}+\frac{2}{e\delta^2\mu}\cdot \pi_{j_1}^{(k)}.	
	\end{eqnarray*}

	Using  $\mu\ge 12\Delta-\pi_{j_1}^{(k)}\ge 11\Delta$ and taking $\delta=1/2$, we have
	$$\EE[\Gamma_1]\le\frac{\pi^{(k)}_{j_1}}{6\Delta}+\frac{2\pi_{j_1}^{(k)}}{e\cdot 1/4\cdot 11\Delta}< \frac{\pi_{j_1}^{(k)}}{2\Delta}. \hspace{20mm}\qedhere$$

\end{proof}

	\begin{claim}\label{claim:2}
		$\EE[\tilde{\Gamma}_2-\Gamma_2]\ge \pi_{j_1}^{(k)}/2.$
	\end{claim}
	
	\begin{proof}
	Let $Y=\sum_{i\neq k} X_{j_2}^{(i)}$. Then
	\begin{eqnarray*}
	\Pr(Y=0)=\Pr(X_{j_2}^{(i)}=0,\forall i\neq k)&=&\prod_{i\neq k} (1-\pi_{j_2}^{i})\\&\ge& 1-\sum_{i\neq k}\pi_{j_2}^{(i)}\\&\ge& 1-(1/2-\pi_{j_2}^{(k)})\\&=&1/2+\pi_{j_2}^{(k)}.	
	\end{eqnarray*}		
	Note that $X_{j_2}^{(k)}$ and $Y$ are independent, hence we have the following,
	\begin{eqnarray*}
		&&\EE[\tilde{\Gamma}_2-\Gamma_2]\\&\ge& \EE[\tilde{\Gamma}_2-\Gamma_2|Y=0,\tilde{X}_{j_2}^{(k)}>0]\Pr(Y=0,\tilde{X}_{j_2}^{(k)}>0)\\
		&\ge& (\pi_{j_1}^{(k)}+\pi_{j_2}^{(k)})(1/2+\pi_{j_2}^{(k)})\cdot\EE[\tilde{\Gamma_2}-\Gamma_2|Y=0,\tilde{X}_{j_2}^{(k)}>0].	
	\end{eqnarray*}
	Further notice that $\tilde{X}_{j_2}^{(k)}$ and $X_{j_2}^{(k)}$ are independent, thus
	\begin{eqnarray*}
		&&\EE[\tilde{\Gamma}_2-\Gamma_2|Y=0,\tilde{X}_{j_2}^{(k)}>0]\\
		&=&\EE[\tilde{\Gamma}_2-\Gamma_2|Y=0,\tilde{X}_{j_2}^{(k)}>0,X_{j_2}^{(k)}=0]\\&&\cdot\Pr(X_{j_2}^{(k)}=0|Y=0,\tilde{X}_{j_2}^{(k)}>0)\\
		&&+\EE[\tilde{\Gamma}_2-\Gamma_2|Y=0,\tilde{X}_{j_2}^{(k)}>0,X_{j_2}^{(k)}>0]\\&&\cdot\Pr(X_{j_2}^{(k)}>0|Y=0,\tilde{X}_{j_2}^{(k)}>0)\\
		&=& \EE[\frac{u_k\tilde{X}_{j_2}^{(k)}}{Y+u_k\tilde{X}_{j_2}^{(k)}}|Y=0,\tilde{X}_{j_2}^{(k)}>0,X_{j_2}^{(k)}=0]\cdot\Pr(X_{j_2}^{(k)}=0)\\
		&&+\EE[\frac{u_k\tilde{X}_{j_2}^{(k)}}{Y+u_k\tilde{X}_{j_2}^{(k)}}-\frac{u_k{X}_{j_2}^{(k)}}{Y+u_k{X}_{j_2}^{(k)}}|Y=0,\tilde{X}_{j_2}^{(k)}>0,X_{j_2}^{(k)}>0]\\&&\cdot\Pr(X_{j_2}^{(k)}>0)\\
		&=& 1-\pi_{j_2}^{(k)}
	\end{eqnarray*}
	Now we have
	\begin{eqnarray*}
		\EE[\tilde{\Gamma}_2-\Gamma_2]\ge (\pi_{j_1}^{(k)}+\pi_{j_2}^{(k)})(1/2+\pi_{j_2}^{(k)})(1-\pi_{j_2}^{(k)})	
	\end{eqnarray*}
	For simplicity, let $\tau=\pi_{j_1}^{(k)}+\pi_{j_2}^{(k)}\in [0,1]$ and $x=\pi_{j_2}^{(k)}$, we have
	$$f(x)=\EE[\tilde{\Gamma}_2-\Gamma_2]-1/2\cdot \pi_{j_1}^{(k)}=\tau(1/2+x)(1-x)-1/2\cdot (\tau-x).$$
	Notice that $f$ is a quadratic function in $x$ whose quadratic term has negative coefficient, therefore if $f(0)\ge 0$ and $f(1)\ge 0$, then for any $x\in [0,1]$ we have $f(x)\ge 0$. It is easy to see that $f(0)= 0$, and $f(1)=1/2\cdot(1-\tau)\ge 0$, hence, the claim is proved.
\end{proof}
	
Given the two claims and the fact that $p_{j_2}\le \Delta p_{j_1}$, $p_{j_2}\EE[\tilde{\Gamma}_2-\Gamma_2]> p_{j_1}\EE[\Gamma_1]$ follows and the lemma is proved. The proof of the two claims is quite involved. 
\end{proof}	

Lemma~\ref{lemma:individual-bound} shows that: Either no transaction has received a total amount of probability that is larger than $12\Delta$, or every transaction receives a total amount of probability at least $1/2$. Note that the two cases are not mutually exclusive. Nevertheless, we show in the following that in both cases, miners will select sufficiently many transactions with very high probability. The proofs of the following lemmas are mathematically involved.

\begin{lemma}\label{le:le12delta}
	If $\sum_{i=1}^m \pi_{j}^{(i)}< 12\Delta$ holds for every transaction $j$, and $\Delta\le m/12$, then the probability that only $k\le \frac{m}{12e\Delta}=\Theta(m/\Delta)$ different transactions are selected by $m$ miners is at most $(1/e)^{\Theta(m)}$.
\end{lemma} 	

Towards the proof of Lemma~\ref{le:le12delta}, we need the following Lemma~\ref{lemma:ineq}.

\begin{lemma}\label{lemma:ineq}
	Let $k,1/\delta\in\mathbb{N}_+$ such that $k\delta=\zeta\in (0,1)$. Let $a_1,a_2,\cdots,a_n$ be $n$ numbers such that $a_i\in [0,\delta]$, $\sum_{i=1}^na_i=1$. Let $\mathcal{A}=\{S|S\subseteq\{1,2,\cdots,n\}, |S|=k\}$. Then  we have
	$$\sum_{S\in\A}(\sum_{j\in S} a_j)^m\le e^k\zeta^{m-k}.$$
\end{lemma}

\begin{proof}
	Without loss of generality we assume that $a_1\ge a_2\ge \cdots a_n$. We define
	$$f(a_1,a_2,\cdots,a_n)=\sum_{S\in\mathcal{A}}(\sum_{j\in S} a_j)^m.$$
	For arbitrary $i<j$ and some small $\epsilon\ge 0$ such that $a_{i-1}\ge a_i+\epsilon$, $a_j-\epsilon\ge a_{j+1}$, we prove in the following that
	$$f(\cdots,a_{i-1},a_i+\epsilon,a_{i+1},\cdots,a_{j-1},a_{j}-\epsilon,a_{j+1},\cdots)\ge f(a_1,a_2,\cdots,a_n)$$
	For simplicity we write $f(a_1,a_2,\cdots,a_n)=f_{i,j}(a_i,a_j)$. We have
	\begin{eqnarray*}
		&& f_{i,j}(a_i+\epsilon,a_j-\epsilon)\\
		&=&\sum_{\substack{S:S\in\A\\ i,j\not\in S}}(\sum_{h\in S} a_h)^m+\sum_{\substack{S:S\in\A\\ i\in S,j\not\in S}}(a_i+\epsilon+\sum_{h\in S,h\neq i} a_h)^m\\&&+\sum_{\substack{S:S\in\A\\ i\not\in S,j\in S}}(a_j-\epsilon+\sum_{h\in S,h\neq j} a_h)^m+\sum_{\substack{S:S\in\A\\ i,j\in S}}(\sum_{h\in S} a_h)^m
	\end{eqnarray*}
	Consider the function $g(x)=x^m$. Due to its convexity, for $x+\epsilon>x>y>y-\epsilon$, we have
	$$g(x+\epsilon)+g(y-\epsilon)\ge g(x)+g(y).$$
	Hence, let $\A'=\{S|S\subseteq\{1,2,\cdots,n\}\setminus\{i,j\}, |S|=k-1\}$, then for any $S'\in \A'$, we have
	\begin{eqnarray*}
	&&(\sum_{h\in S'}a_h+a_i+\epsilon)^m+(\sum_{h\in S'}a_h+a_j-\epsilon)^m\\&\ge& (\sum_{h\in S'}a_h+a_i)^m+(\sum_{h\in S'}a_h+a_j)^m,
	\end{eqnarray*}
	therefore, $f_{i,j}(a_i+\epsilon,a_j-\epsilon)\ge f_{i,j}(a_i,a_j)$.  	
	Now we can iteratively change $a_1,a_2,\cdots,a_n$ into $a_1',a_2',\cdots,a_n'$ such that $a_i'=\delta$ for $1\le i\le 1/\delta$, and $a_i'=0$ otherwise, and get
	$$f(a_1,a_2,\cdots,a_n)\le f(a_1',a_2',\cdots,a_n').$$
	It is not difficult to compute that
	$$f(a_1',a_2',\cdots,a_n')={{1/\delta}\choose{k}}(k\delta)^m\le \frac{(1/\delta)^k}{k!}\cdot \zeta^m\le (\frac{e}{\zeta})^k\zeta^m,$$
	where for the last inequality we make use of the Stirling's approximation that $k!\ge (k/e)^k$. Hence, the lemma is proved.
\end{proof}

Now we come to the proof of Lemma~\ref{le:le12delta}.

\begin{proof}
	Consider the event that at most $k$ different transactions are chosen by $m$ miners. Let $\A=\{S|S\subseteq\{1,2,\cdots,n\}|,
	|S|=k\}$ be the superset of all subsets of cardinality $k$. For any $S\in \A$, let $\sigma_S$ be any mapping that maps $\{1,2,\cdots,m\}$ to $S$, and $\Omega_S$ be the set of all such mappings. Consider the event that miner $i$ selects transaction $\sigma_S(i)$, we know that the event happens with the probability
	$\Pr(\sigma_S)=\prod_{i=1}^m\pi_{\sigma_S(i)}^{(i)}.$
	Taking summation over all possible mappings, the event that only transactions in $S$ are selected is at most
	$$\sum_{\sigma_S\in\Omega_S}\Pr(\sigma_S)=\sum_{\sigma_S\in\Omega_S}\prod_{i=1}^m\pi_{\sigma_S(i)}^{(i)}=\prod_{i=1}^m \sum_{j\in S}\pi^{(i)}_j.$$
	The last equality follows as when we expand $\prod_{i=1}^m \sum_{j\in S}\pi^{(i)}_j$, any summand corresponds to $\prod_{i=1}^m\pi_{\sigma(i)}^{(i)}$ for some mapping $\sigma$. Using the the inequality of arithmetic and geometric means, we have
	{%
	$$\sum_{\sigma_S\in\Omega_S}\Pr(\sigma_S)=\prod_{i=1}^m \sum_{j\in S}\pi^{(i)}_j\le (\frac{\sum_{i=1}^m\sum_{j\in S}\pi_j^{(i)}}{m})^m=(\sum_{j\in S}\frac{\sum_{i=1}^m\pi_j^{(i)}}{m})^m.$$}	
	Taking summation over all possible $S\in \A$, the event that at most $k$ transactions are selected is at most
	$$\sum_{S\in\A}\sum_{\sigma_S\in\Omega_S}\Pr(\sigma_S)\le \sum_{S\in\A}(\sum_{j\in S}\frac{\sum_{i=1}^m\pi_j^{(i)}}{m})^m. $$
	
	Now let $a_j=\frac{\sum_{i=1}^m\pi_j^{(i)}}{m}$. It is easy to see that $\sum_{j=1}^na_j=1$. Furthermore, since $\sum_{i=1}^m\pi_{j}^{(i)}<12\Delta$, we have $a_j\le 12\Delta/m$. We apply Lemma~\ref{lemma:ineq} by taking $\delta=\frac{1}{\lceil m/(12\Delta)\rceil}\le 12\Delta/(m+12\Delta)$, $\zeta=k\delta\le e^{-1}$, we know that
	$$\sum_{S\in\A}\sum_{\sigma_S\in\Omega_S}\Pr(\sigma_S)\le e^k\zeta^{m-k}\le e^{-m/2}. $$
	Hence, the lemma is proved. 	
\end{proof}




Note that if $\Delta>m/12$, $\frac{m}{12e\Delta}<1$. As miners complete at least $1$ transaction, the lemma is trivially true.Now we consider the other case and have the following.

\begin{lemma}\label{le:ge 1}
	If $\sum_{i=1}^m\pi_j^{(i)}\ge 1/2$ holds for every transaction $j$, then the probability that no more than $n/e^2$ transactions are selected is at most $(1/e)^{\Theta(n)}$.	
\end{lemma}

\begin{proof}
	Consider the event that at most $\theta n$ transactions are selected by $m$ miners for some $\theta\in (0,1)$. Note that in this case $m=\sum_{j=1}^n\sum_{i=1}^m\pi_j^{(i)}\ge n/2$, hence $n\le 2m$. Again we define $\A=\{S|S\subseteq\{1,2,\cdots,n\}|,|S|=(1-\theta) n\}$ as the superset of all the subsets of cardinality $(1-\theta n)$. The probability that miner $i$ does {\em not} select any transaction in some $S\in \A$ is $1-\sum_{j\in S}\pi_j^{(i)}$. Given that miners select transactions independently, the probability that all miners do not select transactions in $S$ is
	\begin{eqnarray*}
	\prod_{i=1}^m(1-\sum_{j\in S}p_j^{(i)})&\le& (\frac{m-\sum_{j\in S}\sum_{i=1}^m\pi_j^{(i)}}{m})^m\\ &\le& (\frac{m-\frac{(1-\theta)n}{2}}{m})^m\le e^{-\frac{(1-\theta)n}{2}},	
	\end{eqnarray*}
	where the first inequality follows by inequality of arithmetic and geometric means, the second inequality follows by the fact that $\sum_{i=1}^m\pi_j^{(i)}\ge 1/2$, and the third inequality follows by $(1-1/x)^x\le e^{-1}$ for $x\ge 1$.	Taking the summation over all possible $S\in\A$, the probability that at most $\theta n$ transactions are selected is at most
	\begin{eqnarray*}
	\sum_{S\in\A}\prod_{i=1}^m(1-\sum_{j\in S}p_j^{(i)})&\le& {{n}\choose{\theta n}}e^{-\frac{(1-\theta)n}{2}}\\&\le& \frac{n^{\theta n}}{(\theta n)!}\cdot e^{-\frac{(1-\theta)n}{2}}\\&\le& (\frac{e}{\theta})^{\theta n}\cdot e^{-\frac{(1-\theta)n}{2}}.	
	\end{eqnarray*}
	Taking $\theta=e^{-2}$, simple calculation shows that the right side of the inequality above is $e^{cn}$ for $c\le -0.02$, and the lemma is proved.
\end{proof}

Given Lemma~\ref{le:le12delta} and Lemma~\ref{le:ge 1}, Theorem~\ref{thm:poa-ub} follows directly.


\subsubsection{Finality-duration}
For ease of presentation, we let $s=\min\{c_1m/\Delta,c_2n\}$ (recall Theorem~\ref{thm:poa-ub}). We will characterize finality-duration in terms of $s$. Recall that a miner needs to make two decisions: 
\begin{inparaenum}[(i)]
	\item which transaction to include in the new block,  and
	\item which two previous blocks to refer to.
\end{inparaenum}
The two decisions are independent. In the previous subsection we have discussed (i), and in this subsection we focus on (ii), as this affects how the DAG grows.

It should be clear that since the verification reward of a transaction (block) is evenly distributed among miners who append a block of the same distance to it, a miner always prefers a block with no descendants. At any particular time $t$, we call a block without descendants as a leaf at $t$, and denote by $L_t$ the set of leaves. We are interested in the size of $L_t$. Notice that in the classical blockchain system, $|L_t|$ is 1 since it is a chain. However, in a non-linear model, $|L_t|$ is not necessarily 1. Principally, $|L_t|$ could grow arbitrarily large, but what we will show in this section is that, $|L_t|$ is always bounded when miners are using their equilibrium strategies. In this case, although we are considering a non-linear model, it is ``almost linear", as implied by Theorem~\ref{thm:universal-bound}. Based on this result, we further leverage the techniques from random walk to prove that, for every block, after a delay of $O(s\log s)$ units of time, all blocks will be its descendant (Theorem~\ref{thm:verify}), consequently, if we set $\ell = \Theta(s\log s)$ in our design, every block will be verified by all users, and security follows.  

As we mentioned before, each new block will refer to two leaves in $L_t$. As every block offers the same total amount of verification reward, every leaf appears the same to the miners (unless they are in conflict with previous blocks and then miners will be biased based on the LWD rule). Therefore, a new block will randomly select two leaves to refer to. Assuming leaves are not conflicting with previous blocks, we show that $|L_t|$ will be $O(s)$ in the long run with an extremely high probability. First, it is easy to see that if $|L_t|\le s$, then $L_{t+1}\ge s$ as the $s$ new blocks will be leaves at $t+1$. The following lemma shows that if $|L_t|$ is sufficiently large, then with very high probability it will reduce to $O(s)$ after enough time. 

\begin{lemma}\label{lemma:pre-walk}
	Let $\epsilon$ be an arbitary small constant. If $|L_t|\ge 1/\epsilon^3$ and $|L_t|\ge 4s$, then with sufficiently high probability (at least $1-O(\epsilon)$), $|L_{t+1}|=|L_t|-X+s\le |L_t|-\frac{(1-3\epsilon)s}{2}$, i.e., $L_t$ decreases by at least $\Omega(s)$.
\end{lemma}
\begin{proof}
	Consider an arbitrary $L_t$. For any $\tau_i\in L_t$ and any new block $b$, the probability that $b$ refers to $\tau_i$ is $2/|L_t|$, hence, the probability that none of the $s$ new blocks refer to $\tau_i$ is  $(1-2/|L_t|)^s$, i.e., the probability that an arbitrary $\tau_i$ has descendant(s) at $t+1$ is  $\theta=1-(1-2/|L_t|)^s.$	Let a random variable $X_i$ denote whether the event that  $\tau_i$ has descendant(s) at $t+1$, then $\Pr(X_i=1)=\theta$ and $\Pr(X_i=0)=1-\theta$. Denote by 
	$X=\sum_{i=1}^{|L_t|} X_i$
	the total number of leaves in $L_t$ that has descendant(s) at $t+1$. According to Chernoff bound, with the probability of at least $1-2e^{-\epsilon^2|\EE X|/3}$, $(1-\epsilon)\EE X\le X\le (1+\epsilon)\EE X$. Now we estimate $\EE X$ by verifying
	$$2s-\frac{2s(s-1)}{|L_t|} \le \EE X=|L_t|-|L_t|(1-\frac{2}{|L_t|})^s\le 2s.$$
	Hence, if $|L_t|\ge 1/\epsilon^3$ and $|L_t|\ge 4s$, then  (at least $1-2e^{-1/(3\epsilon)}=1-O(\epsilon)$), $X\ge (1-\epsilon)\EE X\ge 3/2\cdot (1-\epsilon)s$, consequently,
	$$|L_{t+1}|=|L_t|-X+s\le |L_t|-\frac{(1-3\epsilon)s}{2}.\hspace{20mm}\qedhere$$ 
\end{proof}

The above lemma shows that if $|L_t|$ is large, then with high probability $|L_t|$ shall decrease, however, what we are interested in is the probability that $|L_t|\le O(s)$ for all $t\ge 0$. Towards this, we need to cast the problem as a {\em random walk}. Lemma~\ref{lemma:pre-walk} shows that with the probability of $(1-O(\epsilon))^3=1-O(\epsilon)$, $|L_t|$ can decrease by $\frac{3(1-3\epsilon)s}{2}\ge s$, while with probability of at most $O(\epsilon)$, $|L_t|$ can increase by at most $s$. This can be interpreted as a random walk which walks right (increase) by $s$ steps with the probability of $1-O(\epsilon)$, and walks left (decrease) by $s$ steps with the probability of $O(\epsilon)$. The following lemma is proved for a general random walk.

\begin{lemma}[\cite{feller1968introduction}, pp.272]\label{lemma:feller} Consider a random walk starting at $RW_0=0$, $\Pr(RW_{i+1}-RW_i=s)=p$, $\Pr(RW_{i+1}-RW_i=-s)=q$ where $p+q=1$ and $s\in \mathbb{Z}_{>0}$. If $p>q$, then
	$$\lim_{n\to\infty}\Pr(RW_i\ge 0,\,\forall 1\le i\le n)=\frac{p-q}{p}.$$
	If $p< q$, the above limit is $0$.
\end{lemma}

Now we are ready to prove the following theorem.

\begin{theorem}\label{thm:universal-bound}
	Let $\epsilon$ be a small constant such that $s>1/\epsilon^3$.	With very high probability (at least $1-O(\epsilon)$), $|L_t|\le 5s$ for all $t\ge 0$.
\end{theorem}
\begin{proof}
	Recall that $|L_0|=0$.	Let $t^*$ be the smallest time where $|L_{t^*}|\ge 4s$, then $|L_{t^*}|\le 5s$. Now we take $t^*$ as a starting time, $|L_{t^*}|$ as a starting point and take the random walk interpretation. Using Lemma~\ref{lemma:feller}, we have that
	\begin{eqnarray*}
	\lim_{n\to\infty}\Pr(|L_{t}|\le |L_{t^*}|, \forall 1\le t\le n)&\le& \frac{1-O(\epsilon)-O(\epsilon)}{1-O(\epsilon)}\\&=&1-O(\epsilon).
	\end{eqnarray*}
	Therefore, the probability that $|L_t|$ is bounded by $5s$ for all $t\ge 0$ is at least $1-O(\epsilon)$.
\end{proof}

\begin{lemma}\label{lemma: descendant}
	Let $\epsilon$ be a small constant such that $s>1/\epsilon^3$.	For any transaction at $t$ that is not in conflict with prior transactions, with sufficiently high probability (at least $1-O(\epsilon)$) every block appended at or after $t+O(s\log s)$ will be its descendant.
\end{lemma}

\begin{proof}
	According to Theorem~\ref{thm:universal-bound}, we focus on the event that $|L_t|\le 5s$ for all $t\ge 0$, which happens with $1-O(\epsilon)$ probability.	
	
	For $h\ge t$, let $\Psi_h$ be the subset of blocks in $L_{h}$ which has a directed path from some fixed block $\tau_0\in L_t$, which is a random subset. Let $\psi_h=\EE(|\Psi_h|)$. Consider $L_{h+1}$. For any block $\tau_i\in L_{h+1}$, let $X_i$ be a binary random variable indicating whether $\tau_i$ refers to some block in $\Psi_h$, and hence admits a directed path from $\tau_0$. Then we know 
	\begin{eqnarray*}
	\Pr(X_i=1)&=&\frac{{{|\Psi_h|}\choose{2}}+|\Psi_h|(|L_h|-|\Psi_h|)}{{{|L_h|}\choose{2}}}\\&=&\frac{|\Psi_h|(2|L_h|-|\Psi_h|-1)}{|L_h|(|L_h|-1)}.
	\end{eqnarray*}
	
	We consider $|\Psi_{h+1}|$. It is obvious that if $|\Psi_h|=|L_h|$, then every block in $L_{h+1}$ refers to some block in $\Psi_h$ and thus admits a directed path from $\tau_0$, hence, $|L_{h+1}|=|\Psi_{h+1}|$, and similarly we have $|L_{h+j}|=|\Psi_{h+j}|$ for all $j\ge 1$. 	Otherwise, we assume $1\le |\Psi_h|\le |L_h|-1$. Then $2|L_h|-|\Psi_h|-1\ge |L_h|$, and we have
	$$\EE(X_i)=\EE\left(\frac{|\Psi_h|(2|L_h|-|\Psi_h|-1)}{|L_h|(|L_h|-1)}\right)\ge \frac{\psi_h}{|L_h|-1}.$$
	Note that $|\Psi_{h+1}|=\sum_i X_i$. It is easy to calculate that $$\psi_{h+1}=\EE(|\Psi_{h+1}|)\ge \psi_h\left(1+\frac{1}{|L_h|-1}\right).$$
	
	This means, starting from $\psi_t=1$, for each $\psi_h$ where $h\ge t$, either $\psi_h=|L_h|$ and thus $\psi_{h'}=|L_{h'}|$ for all $h'\ge h$, or $\psi_{h+1}\ge \left(1+\frac{1}{|L_h|-1}\right)\psi_{h}$. Since $|L_h|\le 5s$,  $\psi_h$ increases  sufficiently close to $|L_h|\le 5s$ when $h\ge t+O(s\log s)$, and the theorem is proved. 	
\end{proof}

Given the above lemma, if we set $\ell$, the verification depth to be $\ell\ge O(s\log s)$, then any transaction at $t$ will be verified by all the users after $O(s\log s)$ units of time with high probability. The following theorem is thus true.
\begin{theorem}\label{thm:verify}
	If $s>1/\epsilon^3$ and $\ell\ge O(s\log s)$, then with probability of at least $1-O(\epsilon)$, any transaction at $t$ will be verified by all the users after $O(s\log s)$ units of time.
\end{theorem}

\para{Remark.} Recall that the scalability of the system increases as $\Delta$ increases, while $s=\min\{c_1m/\Delta,c_2n\}$, and hence the finality-duration $O(s\log s)$ decreases as $\Delta$ increases. Theorem~\ref{thm:verify} shows trade-off between the scalability and finality-duration. 



\section{Conclusion}\label{sec-conclusion}
We provide the first systematic analysis on blockchain systems with respect to three major parameters, verification, scalability, and finality-duration. We establish an impossibility result showing no blockchain system can simultaneously achieve the three properties. We complement the existing blockchain systems by establishing the first NLB that achieves both full verification and scalability. We also reveal, for the first time, the trade-off between scalability and finality-duration in NLB. It is not clear whether a better trade-off exists or not. 

\clearpage

\balance
\bibliographystyle{abbrv}
\bibliography{reference}

\begin{thebibliography}{10}

\bibitem{ali2019applications}
M.~S. Ali, M.~Vecchio, M.~Pincheira, K.~Dolui, F.~Antonelli, and M.~H. Rehmani.
\newblock Applications of blockchains in the internet of things: {A}
  comprehensive survey.
\newblock {\em {IEEE} Communications Surveys Tutorials}, 21(2):1676--1717,
  2019.

\bibitem{azouvi2018betting}
S.~Azouvi, P.~McCorry, and S.~Meiklejohn.
\newblock Betting on blockchain consensus with fantomette.
\newblock {\em CoRR}, abs/1805.06786, 2018.

\bibitem{bano2019sok}
S.~Bano, A.~Sonnino, M.~Al{-}Bassam, S.~Azouvi, P.~McCorry, S.~Meiklejohn, and
  G.~Danezis.
\newblock Sok: Consensus in the age of blockchains.
\newblock In {\em Proceedings of the 1st {ACM} Conference on Advances in
  Financial Technologies, {AFT} 2019}, pages 183--198, Zurich, 2019. {ACM}.

\bibitem{boyen2017blockchain}
X.~Boyen, C.~Carr, and T.~Haines.
\newblock Blockchain-free cryptocurrencies.
\newblock 2017.

\bibitem{chaum1983blind}
D.~Chaum.
\newblock Blind signatures for untraceable payments.
\newblock In {\em Advances in Cryptology: Proceedings of {CRYPTO} '82}, pages
  199--203, California, 1982.

\bibitem{conti2018survey}
M.~Conti, S.~K. E, C.~Lal, and S.~Ruj.
\newblock A survey on security and privacy issues of bitcoin.
\newblock {\em {IEEE} Communications Surveys Tutorials}, 20(4):3416--3452,
  2018.

\bibitem{feller1968introduction}
W.~Feller.
\newblock {\em An introduction to probability theory and its applications:
  volume I}, volume~3.
\newblock John Wiley \& Sons, New York, 3rd edition, 1968.

\bibitem{francca2015homomorphic}
B.~Fran{\c{c}}a.
\newblock Homomorphic mini-blockchain scheme, 2015.

\bibitem{garay2020sok}
J.~A. Garay and A.~Kiayias.
\newblock Sok: {A} consensus taxonomy in the blockchain era.
\newblock In {\em Topics in Cryptology - {CT-RSA} 2020 - The Cryptographers'
  Track at the {RSA} Conference 2020, Proceedings}, volume 12006 of {\em
  Lecture Notes in Computer Science}, pages 284--318, San Francisco, 2020.
  Springer.

\bibitem{khalilov2018survey}
M.~C.~K. Khalilov and A.~Levi.
\newblock A survey on anonymity and privacy in bitcoin-like digital cash
  systems.
\newblock {\em {IEEE} Communications Surveys Tutorials}, 20(3):2543--2585,
  2018.

\bibitem{liu2019survey}
Z.~Liu, N.~C. Luong, W.~Wang, D.~Niyato, P.~Wang, Y.~Liang, and D.~I. Kim.
\newblock A survey on applications of game theory in blockchain.
\newblock {\em CoRR}, abs/1902.10865, 2019.

\bibitem{miers2013zerocoin}
I.~Miers, C.~Garman, M.~Green, and A.~D. Rubin.
\newblock Zerocoin: Anonymous distributed e-cash from bitcoin.
\newblock In {\em 2013 {IEEE} Symposium on Security and Privacy, {SP} 2013},
  pages 397--411, California, 2013. {IEEE} Computer Society.

\bibitem{nakamoto2008bitcoin}
S.~Nakamoto.
\newblock Bitcoin: A peer-to-peer electronic cash system, 2008.

\bibitem{nash1951non}
J.~Nash.
\newblock Non-cooperative games.
\newblock {\em Annals of mathematics}, pages 286--295, 1951.

\bibitem{popov2016tangle}
S.~Popov.
\newblock The tangle.
\newblock {\em cit. on}, page 131, 2016.

\bibitem{popov2019equilibria}
S.~Popov, O.~Saa, and P.~Finardi.
\newblock Equilibria in the tangle.
\newblock {\em Computers \& Industrial Engineering}, 136:160--172, 2019.

\bibitem{rocket2018snowflake}
T.~Rocket.
\newblock Snowflake to avalanche: A novel metastable consensus protocol family
  for cryptocurrencies.
\newblock 2018.

\bibitem{salman2018security}
T.~Salman, M.~Zolanvari, A.~Erbad, R.~Jain, and M.~Samaka.
\newblock Security services using blockchains: {A} state of the art survey.
\newblock {\em {IEEE} Communications Surveys Tutorials}, 21(1):858--880, 2019.

\bibitem{sander1999auditable}
T.~Sander and A.~Ta{-}Shma.
\newblock Auditable, anonymous electronic cash extended abstract.
\newblock In {\em Advances in Cryptology - {CRYPTO} '99, 19th Annual
  International Cryptology Conference, Proceedings}, volume 1666 of {\em
  Lecture Notes in Computer Science}, pages 555--572, California, 1999.
  Springer.

\bibitem{sasson2014zerocash}
E.~B. Sasson, A.~Chiesa, C.~Garman, M.~Green, I.~Miers, E.~Tromer, and
  M.~Virza.
\newblock Zerocash: Decentralized anonymous payments from bitcoin.
\newblock In {\em 2014 {IEEE} Symposium on Security and Privacy, {SP} 2014},
  pages 459--474, California, 2014. IEEE.

\bibitem{SompolinskyLZ16}
Y.~Sompolinsky, Y.~Lewenberg, and A.~Zohar.
\newblock {SPECTRE:} {A} fast and scalable cryptocurrency protocol.
\newblock {\em {IACR} Cryptology ePrint Archive}, 2016:1159, 2016.

\bibitem{sompolinsky2015secure}
Y.~Sompolinsky and A.~Zohar.
\newblock Secure high-rate transaction processing in bitcoin.
\newblock In {\em Financial Cryptography and Data Security - 19th International
  Conference, {FC} 2015, , Revised Selected Papers}, volume 8975 of {\em
  Lecture Notes in Computer Science}, pages 507--527, Puerto Rico, 2015.
  Springer.

\bibitem{sompolinsky2020phantom}
Y.~Sompolinsky and A.~Zohar.
\newblock Phantom, ghostdag, 2020.

\bibitem{ethereum20}
{The Ethereum Team}.
\newblock On sharding blockchains, 2019.

\bibitem{harmony}
{The Harmony Team}.
\newblock Harmony - technical whitepaper, 2018.

\bibitem{zilliqa}
{The Zilliqa Team}.
\newblock The zilliqa technical whitepaper, 2017.

\bibitem{Solidity_ether}
N.~P. Triantafyllidis and T.~Oskar~van Deventer.
\newblock Developing an ethereum blockchain application.
\newblock 2016.

\bibitem{tschorsch2016bitcoin}
F.~Tschorsch and B.~Scheuermann.
\newblock Bitcoin and beyond: A technical survey on decentralized digital
  currencies.
\newblock {\em {IEEE} Communications Surveys \& Tutorials}, 18(3):2084--2123,
  2016.

\bibitem{vukolic2015quest}
M.~Vukolic.
\newblock The quest for scalable blockchain fabric: Proof-of-work vs. {BFT}
  replication.
\newblock In {\em Open Problems in Network Security - {IFIP} {WG} 11.4
  International Workshop, iNetSec 2015, Revised Selected Papers}, volume 9591
  of {\em Lecture Notes in Computer Science}, pages 112--125, Zurich, 2015.
  Springer.

\bibitem{xu2017epbc}
L.~Xu, L.~Chen, Z.~Gao, S.~Xu, and W.~Shi.
\newblock {EPBC:} efficient public blockchain client for lightweight users.
\newblock In {\em Proceedings of the 1st Workshop on Scalable and Resilient
  Infrastructures for Distributed Ledgers, SERIAL@Middleware 2017}, pages
  1:1--1:6, Nevada, 2017. {ACM}.

\bibitem{yang2019integrated}
R.~Yang, F.~R. Yu, P.~Si, Z.~Yang, and Y.~Zhang.
\newblock Integrated blockchain and edge computing systems: {A} survey, some
  research issues and challenges.
\newblock {\em {IEEE} Communions Surveys \& Tutorials}, 21(2):1508--1532, 2019.

\end{thebibliography}
\clearpage

\end{document}